\documentclass[11pt]{article}

\usepackage{amsmath, amsthm, amssymb}
\usepackage{inputenc, graphicx, color, algorithm, algorithmic}
\usepackage{titling}
\thanksmarkseries{arabic}

\newtheorem{thm}{Theorem}[section]
\newtheorem{cor}[thm]{Corollary}
\newtheorem{prop}[thm]{Proposition}
\newtheorem{lem}[thm]{Lemma}

\theoremstyle{definition}
\newtheorem{defi}[thm]{Definition}
\newtheorem{rem}[thm]{Remark}
\newtheorem{ex}[thm]{Example}
\numberwithin{equation}{section}
\numberwithin{algorithm}{section}
\newcommand{\rowsp}{\mathrm{rowsp}}
\newcommand{\rk}{\mathrm{rk}}

\newcommand{\F}{\mathbb{F}}

\newcommand{\N}{\mathbb{N}}
\newcommand{\G}{\mathcal{G}_q(k,n)}
\newcommand{\U}{\mathcal{U}}
\newcommand{\cM}{\mathcal{M}}

\newcommand{\cG}{\mathcal{G}}
\newcommand{\cH}{\mathcal{H}}
\newcommand{\cS}{\mathcal{S}}

\newcommand{\Vvs}{\mathcal{V}}
\newcommand{\Uvs}{\mathcal{U}}
\newcommand{\Rvs}{\mathcal{R}}
\newcommand{\Evs}{\mathcal{E}}
\newcommand{\eavg}{\mbox{$e_{\rm{avg}}$}}
\newcommand{\C}{\mathcal{C}}

\newcommand{\GL}{\mathrm{GL}}

\newcommand{\Rs}{\mbox{$R_{\mathcal S}$}}
\newcommand{\Rh}{\mbox{$R_{\mathcal H}$}}
\newcommand{\eh}{\mbox{$e_{\mathcal H}$}}
\newcommand{\Gaussian}[2]{\mbox{$\Big[\mbox{${#1}\atop{#2}$}\Big]$}}
\newcommand{\rs}{\mathrm{rowsp}}

\title{Symbol Erasure Correction Capability\\ of Spread Codes}

\author{Heide Gluesing-Luerssen\thanks{HGL is with the Department of Mathematics, University of Kentucky,
Lexington KY 40506-0027, USA; heide.gl@uky.edu. She was partially supported by the grant \#422479 from the Simons Foundation.}\ \ and Anna-Lena Horlemann-Trautmann\thanks{ALHT is with the Faculty of Mathematics and Statistics, University of St.\ Gallen, St.\ Gallen, Switzerland; anna-lena.horlemann@unisg.ch.}}

\date{}

\begin{document}

\maketitle

\abstract{
We consider data transmission over a network where each edge is an erasure channel and where the inner nodes transmit a random linear combination of their incoming information. We distinguish two channel models in this setting, the row and the column erasure channel model. For both models we derive the symbol erasure correction capabilities of spread codes and compare them to other known codes suitable for those models. Furthermore, we explain how to decode these codes in the two channel models and compare their decoding complexities. The results show that, depending on the application and the to-be-optimized aspect, any combination of codes and channel models can be the best choice.
}

\section{Introduction}

Network coding in general, and random (or non-coherent) network coding in particular, has received much attention in the last decade. Subspace codes, first introduced in \cite{ko08}, are a class of codes well suited for error correction in random network coding. By definition, they are sets of subspaces of some given vector space of dimension $n$ over the finite field $\F_q$. One of the most studied families of subspace codes are \emph{spread codes} (or simply \emph{spreads}), objects that had been studied in finite geometry for a long time, without the application to coding theory.

In the classical setup, as used in \cite{ko08}, one considers a network whose edges are $q$-ary symmetric channels, i.e., where symbols from $\F_q$ might be changed into other symbols of $\F_q$ during transmission. In this paper however, we focus on networks whose edges are erasure channels, i.e., where symbols are either unchanged or erased during transmission. This scenario has been studied significantly less than the classical setup, but some works exist, see e.g.\ \cite{sk13}. In \cite{sk13} the authors define \emph{hybrid codes} to correct both symbol erasures and classical errors. These codes are defined as a composition of Reed-Solomon and subspace codes.

In this work we investigate the performance of spread codes over an erasure-only network channel. More precisely, we compare the symbol erasure correction capability of spread codes in two different network channel models, the \emph{row erasure channel} and the \emph{column erasure channel}. Furthermore, we compare the results to the erasure correction capability of hybrid codes in the same scenario. As a next step, we give decoding algorithms for the various codes and channel models and derive their computational complexities. For most of the paper we assume that the network channel is deletion-free, i.e., that no rank deficiencies occur during transmission. However, we also consider the case with deletions in the end of the paper.

The paper is structured as follows. We start with preliminaries about finite fields, subspace codes and rank-metric codes in Section \ref{S-Prelim}. In Section \ref{S-Models} we explain the two network channel models we are going to investigate, namely the row erasure channel and the column erasure channel model. The first main results, the symbol erasure correction capabilities of spread codes in the two channel models, are derived in Section \ref{sec:spreads}. In Section \ref{S-Hybrid} we compare these numbers to the erasure correction capability of hybrid codes. In Section \ref{sec:decoding} we show how to decode spread codes in the two channel models. We derive the corresponding decoding complexities and compare them to the decoding performance of hybrid codes. Finally, in Section \ref{sec:simul}, we consider the column erasure channel with deletions. We again derive the symbol erasure correction capability and a decoding algorithm for spread codes and compare this performance to the one of hybrid codes. We conclude this work in Section \ref{sec:conclusion}.


\section{Preliminaries}\label{S-Prelim}

We first state some well-known preliminary results about finite fields. Most of the results, as well as their proofs and further information on finite fields, can be found e.g. in \cite{li94}.

Let $q$ be a prime power and $\F_q$ be the finite field with $q$ elements. The set of invertible elements is denoted by $ \F_q^*:=\F_q\backslash\{0\}$.  Let $p(x)=\sum_{i=0}^{k-1}p_i x^i+x^k \in \F_q[x]$ be a monic irreducible polynomial and let $\alpha \in \F_{q^k}$ be a root of it. Then
$$ \F_{q^k} \cong \F_q[\alpha] .$$
Throughout this paper, we realize the field~$\F_{q^k}$ as $\F_q[\alpha]$, if not noted differently.

For a monic polynomial $p(x) =\sum_{i=0}^{k-1}p_i x^i+x^k \in \F_q[x]$ of degree $k$ the matrix
\begin{equation}\label{e-P}
   P =\begin{pmatrix} 0&0&\cdots& 0&-p_0\\ 1&0&\cdots&0&-p_1\\0&1&\cdots&0&-p_2\\ & &\ddots  & &\vdots\\ 0&0&\cdots&1&-p_{k-1}\end{pmatrix}
\end{equation}
is the \emph{companion matrix} of $p(x)$.
If $p(x)$ is irreducible and $\alpha$ is a root of it, then $\F_q[P] \cong \F_q[\alpha]$, i.e., $\F_q[P]$ is a field of size $q^k$. Hence all nonzero elements of $\F_q[P]$ are invertible, i.e., they have rank $k$. Naturally, the same holds for the transposed matrices, i.e., for $\F_q[P^\top]$.

We have the natural field isomorphism
\begin{equation}\label{e-phi}
\phi:\F_{q^k} \cong \F_q[P],\quad \sum_{i=0}^{k-1} v_i\alpha^i\longmapsto \sum_{i=0}^{k-1}v_i P^i    ,
\end{equation}
and the vector space isomorphism
\begin{equation}\label{e-psi}
    \psi:\F_q^k\longrightarrow\F_{q^k},\quad (v_0,\ldots,v_{k-1})\longmapsto \sum_{i=0}^{k-1}v_i\alpha^i.
\end{equation}
We also extend $\psi$ to
\begin{equation}\label{e-barpsi}
 \bar\psi : \F_q^{\ell\times k}  \longrightarrow \F_{q^k}^\ell, \quad
 \begin{pmatrix}  a_{11} & \dots & a_{1k} \\
 \vdots && \vdots\\
 a_{\ell 1} & \dots & a_{\ell k}
                           \end{pmatrix}   \longmapsto   \Big(\psi(a_{11},\dots,a_{1k}), \dots, \psi(a_{\ell 1},\dots,a_{\ell k})\Big) ,
\end{equation}
which we may use for various values of~$\ell$.

Note that $\psi(\mathbf v P^\top)=\psi(\mathbf v)\alpha$ for all row vectors $\mathbf v\in\F_q^k$ and $\psi(P\mathbf u)=\alpha\psi(\mathbf u)$ for all column vectors $\mathbf u\in\F_q^k$.
As a consequence, for all $s\in\N_0$,
\begin{equation}\label{e-psiP}
\bar\psi\left((P^s)^\top\right)=(\alpha^s,\ldots,\alpha^{s+k-1})  .
\end{equation}

\medskip

The \emph{Grassmannian variety} $\G$ is the set of all $k$-dimensional subspaces of the $n$-dimensional vector space $\F_q^n$.
It is a metric space with respect to the subspace distance $d_S$, defined as
\[
    d_S(\Uvs, \Vvs) := \dim (\Uvs + \Vvs) - \dim (\Uvs \cap \Vvs) = 2k - 2\dim(\Uvs \cap \Vvs)
\]
for all $\Uvs, \Vvs \in \G$, see e.g.\ \cite{ko08}.

\begin{defi}
A \emph{constant dimension (subspace) code of dimension~$k$ and length~$n$} is a subset $\C\subseteq\G$. The minimum subspace distance of $\C$ is defined as
$$ d_S(\C) := \min \{d_S(\Uvs, \Vvs) \mid \Uvs,\, \Vvs \in \C,\, \Uvs \neq \Vvs\}  . $$
\end{defi}

We can represent a subspace $\Uvs \in \G$ by a basis matrix $U \in \F_q^{k\times n}$ in the sense that $\rs(U) = \Uvs$, where $\rs(U)$ denotes the row space of $U$.
This representation is not unique.
However one can determine a unique matrix representation for the elements of $\G$, e.g., by choosing the basis matrices in reduced row echelon form (RREF).

Subspace codes were originally introduced in \cite{ko08} for error-correction in random (or non-coherent) network coding. In \cite{ko08} the authors consider a single-source multicast network channel, where every edge of the network can be thought of as a $q$-ary symmetric channel and the inner nodes of the network send a random linear combination of their incoming information along the outgoing edges.
They model this as the \emph{operator channel}, which takes as input a $k$-dimensional vector space $\U \in \G$ and outputs a received word of the form
$$ \Rvs = \bar \U \oplus \Evs ,$$
where $\bar\U $ is a subspace of $\U$ and $\Evs$ is the error space such that $\Evs \cap \Uvs =\{\mathbf{0}\}$. In practice, the source sends a basis of $\U$ along its outgoing edges (one vector per edge) and the receiver gets a set of vectors generating $\Rvs$.

There are two types of errors that can be observed at the receiver: \emph{deletions}, which correspond to the dimension losses from $\U$ to $\bar \U$, and \emph{insertions}, which correspond to the dimension gains due to the error space $\Evs$. A constant dimension code with minimum subspace distance $\delta$ can correct up to $(\delta-1)/2$ errors (deletions $+$ insertions).
For more information on the operator channel the reader is referred to \cite{ko08}.

One of the most studied families of constant dimension codes are spread codes.
\begin{defi}
A \emph{spread (code)} in $\G$ is a subset of  $\G$ such that all elements intersect pairwise trivially and their union  covers the whole vector space $\F_q^n$.
\end{defi}

Spreads are well-known geometrical objects.
A simple counting argument shows that they exist if and only if $k\mid n$, in which case they have $(q^n-1)/(q^k-1)$ elements.
As a constant dimension code they have minimum subspace distance $2k$.

The following construction for spread codes (see e.g.\ \cite{ma08p}) will be considered later in this paper.

\begin{defi}\label{D-DesSpread}
Let  $P\in\GL_k(q)$ as in~\eqref{e-P} be the companion matrix of a monic irreducible polynomial in $\F_q[x]$ of degree $k$.
Fix $m\in\N$ and set $n=mk$.
Then the code
\[
  \mathcal S_q(m,k,P)=\{\text{rowsp}(G)\mid G\in \cM\}\subseteq\G,
\]
where
\begin{equation*}
 \cM=\Big\{(0_{k\times k}\mid\ldots\mid 0_{k\times k}\mid I_k\mid B_{i+1}\mid \ldots\mid B_m)\,\Big|\, i=1,\ldots,m,\,B_i\in\F_q[P]\Big\},
\end{equation*}
is a spread code.
We call any spread of this form a \emph{Desarguesian spread}.
Similarly, using~$P^\top$ instead of~$P$ leads to the spread code $\cS_q(m,k,P^\top)$.
\end{defi}

Note that each matrix in~$\cM$ is in reduced row echelon form.
The isomorphism~$\phi$ from \eqref{e-phi} implies that the Desarguesian spread $\mathcal S_q(m,k,P)$ is isomorphic to the Grassmannian
$\cG_{q^k} (1,m)$ via
\begin{align}\label{e-DesGrass}
   \cG_{q^k} (1,m)  &\longrightarrow\mathcal S_q(m,k,P) \nonumber\\
   \rs(u_1,\ldots,u_m)   &\longmapsto\rs\Big(\phi(u_1)\mid\ldots\mid\phi(u_m)\Big).
\end{align}
Analogously we get an isomorphism from $\cG_{q^k} (1,m)$ to $\mathcal S_q(m,k,P^\top)$.
This justifies the terminology \emph{Desarguesian} as the above defined spreads are isomorphic to $\cG_{q^k}(1,m)$, which in turn is a representation of a
Desarguesian $(k-1)$-spread as known in finite geometry. This fact is used in Section \ref{sec:decoding}.

\medskip

Another class of codes, related to subspace codes, are rank-metric codes.
They are defined as subsets of the matrix space $\F_q^{m\times n}$, which forms a metric space with the \emph{rank distance} $d_R$, defined as
$$ d_R(U,V) := \rk (U-V) $$
for all $U,V \in \F_q^{m\times n}$.

\begin{defi}
An $m\times  n$ \emph{rank-metric code} is a subset $C\subseteq \F_q^{m\times n}$.
The minimum rank distance of $C$ is defined as
$$ d_R(C) := \min \{d_R(U,V) \mid U,V \in C , U \neq V\}  . $$
\end{defi}

Rank-metric codes can be used for correcting various error and erasure types. In this paper we will focus on row and column erasures, which means that a complete row, respectively column,
is erased in a matrix in $\F_q^{m\times n}$.

The following result can also be found in a more general version in \cite{ri04p}. For completeness
 we also give a proof of the result.

\begin{lem}\label{L-RP}
Let $C\subseteq \F_q^{k\times n}$ be a linear
rank-metric code of minimum rank distance $k$.
Then any combination of~$r$ row erasures and~$c$ column erasures can be decoded as long as $r+c<k$.
\end{lem}

\begin{proof}
Assume without loss of generality that the first~$r$ rows and the first~$c$ columns of the matrix $A\in C$ have been erased.
Then a $(k-r)\times(n-c)$-matrix~$\hat{A}$ is received.
Suppose there are two matrices $A_1,\,A_2\in C$ with the same submatrix~$\hat{A}$ in the lower right corner.
Then the difference $A_1-A_2$ is zero in that $(k-r)\times(n-c)$-submatrix and thus $d_R(A_1,A_2)\leq r+c$.
Since $r+c<k$, we conclude $A_1=A_2$.
\end{proof}

One of the most studied families of rank-metric codes are Gabidulin codes \cite{de78,ga85a}.

\begin{defi}
Let $k\geq\ell$ and $\beta_1,\dots,\beta_\ell \in \F_{q^k}$ be linearly independent over $\F_q$. Then the $\F_{q^k}$-linear subspace $C \subseteq \F_{q^k}^\ell$ with generator matrix
$$G= \left( \begin{array}{ccccc} \beta_1 & \dots & \beta_\ell  \\ \beta_1^{q} & \dots & \beta_\ell^q  \\  \vdots && \vdots \\ \beta_1^{q^{s-1}} & \dots & \beta_\ell^{q^{s-1}}    \end{array}\right)$$
is called a \emph{Gabidulin code} of length $\ell$ and $\F_{q^k}$-dimension $s$.
The matrix representation $\bar \psi^{-1} (C)$ is a linear rank-metric code of $\F_{q}$-dimension $ks$ in $\F_q^{\ell \times k}$. We will use the name \emph{Gabidulin code} for both representations.
\end{defi}

Gabidulin codes are optimal in the sense that their minimum rank distance $d_R$ achieves the Singleton bound $d_R = \ell-s+1$. For more information the interested reader is referred to \cite{ga85a}.

\begin{prop}\label{P-FPGab}
Let $p(x)=\sum_{i=0}^{k-1}p_ix^i + x^k \in \F_q[x]$ be irreducible, $\alpha$ a root of $p(x)$ and $P$ the companion matrix as in~\eqref{e-P}.
\begin{enumerate}
\item
$\bar\psi(\F_q[P^\top])$ is the Gabidulin code in $\F_{q^k}^k$ with generator matrix
$G= ( 1 \; \alpha \; \dots \; \alpha^{k-1})$.
It thus has length~$k$, $\F_{q^k}$-dimension~$1$, and minimum rank distance~$k$.
\item
For any $S\in \F_q^{(k-r)\times k}$ of full rank
 the code $\bar\psi(S\F_q[P^\top])$  is the Gabidulin code in $\F_{q^k}^{k-r}$ with generator matrix
$GS^\top$.
\end{enumerate}
\end{prop}

\begin{proof}
\begin{enumerate}
\item
The first statement follows from \eqref{e-psiP}.
\item
Let us collect some simple properties of the maps~$\psi$ and~$\bar\psi$.
For a matrix $A\in\F_q^{k\times k}$ denote the rows by $A_1,\ldots,A_k$.
Then one easily verifies that for any $\mathbf v\in\F_q^k$ and $A\in\F_q^{k\times k}$
one has $\psi(\mathbf vA)=\sum_{i=1}^k v_i\psi(A_i)$.
From this one obtains the identity $\bar\psi(SA)=\bar\psi(A)S^\top$
for any $A\in\F^{k\times k}$.
From 1.\
 along with the $\F_q$-linearity of~$\bar\psi$ it follows that  $\bar\psi(S\F_q[P^\top])$  is the Gabidulin code with generator matrix $GS^\top$.
 \qedhere
\end{enumerate}
\end{proof}

Naturally, the above implies that $\F_q[P]$ is also a rank-metric code in $\F_q^{k\times k}$ with minimum rank distance $k$, and the following is immediate with
Lemma \ref{L-RP}.

\begin{cor}\label{C-Column}
Let $P\in \F_q^{k\times k}$ be the companion matrix of an irreducible monic polynomial in $\F_q[x]$ of degree $k$. Let $r\in\{0,\ldots,k-1\}$.
Then the rank-metric codes $\F_q[P]$ and $\F_q[P^\top]$ can decode any $r$ row erasures and $k-r-1$ column erasures.
\end{cor}

One of the most commonly used relationships between rank-metric and subspace codes is the following. From any rank-metric code $C\subseteq \F_q^{k\times (n-k)}$ we can construct a constant dimension code $\C \subseteq \G$ via the \emph{lifting} operation:
\begin{equation}\label{eq:lifting}
\C = \mathrm{lift}(C) := \{\rs (I_k \mid U) \mid U \in C\} .
\end{equation}
If $C$ has minimum rank distance $d_R$ one can easily see that $\C$ has minimum subspace distance $d_S=2d_R$.


\section{Two Models for Symbol Erasures in Linear Random Network Coding}\label{S-Models}

As for the operator channel, we consider the classical single-source multicast network coding setting, where we allow the inner nodes to randomly linearly combine and forward their incoming information. However, we now assume that the edges of the network are erasure channels, instead of $q$-ary symmetric channels. To distinguish from other notions of erasures in network channels, we speak of \emph{symbol erasures}, which are defined as the erasure of a single entry in a vector sent along any edge.

In order to model symbol erasures, we expand the underlying alphabet from $\F_q$ to $\F_q \cup \{?\}$, where $?$ denotes a
symbol erasure.
\begin{defi}
The (commutative) operations with~$?$ are defined as
\begin{equation}\label{e-QuestMark}
    0*? = 0 , \quad x*? = ?, \quad  \quad y+?  = ?,\quad \textnormal{and } ?+?=?=?*?
\end{equation}
for $x\in \F_q^*$ and $y \in \F_q$.
\end{defi}

Since every $k$-dimensional subspace $\Uvs \leq \F_q^n$ can be described by a basis matrix $U\in \F_q^{k\times n}$, we can model the channel as a matrix channel,
instead of a vector space channel (as the operator channel). It turns out that for our purposes the matrix description is advantageous over the subspace description.

For simplicity we first describe the erasure-free channel model. The input of the channel is a basis matrix $U\in \F_q^{k\times n}$ of some vector space $\Uvs \in \G$. The output is
\begin{equation*}
    R = A U \in \F_q^{k\times n},
\end{equation*}
where $A\in \F_q^{k\times k}$ is the representation of the random operations of the inner nodes of the network channel. Clearly, if $A$ has full rank, $R$ is simply another matrix representation of the subspace $\Uvs$. If $A$ does not have full rank, then $\rs(R)$ is a subspace of $\Uvs$. This rank deficiency is called a \emph{deletion}.

\medskip

We now allow symbol erasures to happen along the edges of the network.
In the (random) network coding literature two models have been proposed to deal with symbol erasures.
First, K\"otter/Kschi\-schang~\cite{ko08} proposed that one can use the operator channel and consider a vector with an erasure as faulty and ignore it at the receiving node,
see \cite[p.~3581]{ko08}.
This could possibly lead to a deletion, i.e., a dimension loss of the codeword.
The K\"otter-Kschischang model with only symbol erasures can thus be described as follows.

\medskip

\textbf{Row Erasure Channel Model.}
Define the \emph{row deletion operator} $\rho$ on the matrix space $(\F_q\cup\{?\})^{k\times n}$ to delete every row of the matrix that contains an erasure.
If the channel takes as input a matrix $U \in \F_q^{k\times n}$, we may write the output as
\begin{equation*}
   \hat{R}=\rho(AU+E )   \ \in \F_q^{(k-r)\times n} ,
\end{equation*}
where $E\in \{0, ?\}^{k\times n}$ is the \emph{symbol erasure matrix} such that $r$ rows contain an erasure, and $A\in\F_q^{k\times k}$ represents the channel operation matrix.
The assumption that the channel ignores a partially erased vector right at the receiving node is taken care of by this model by taking~$A$ suitably (e.g., if the last~$r$ rows are erased, choose~$A$ as a block diagonal matrix with an $r\times r$-identity in the last block and such that the first block represents the downstream channel operations on the non-erased vectors).
Clearly, $\rowsp(\hat{R})$ is a subspace of $\rowsp(U)$. Note that, as in the erasure-free case, $A$ does not necessarily have full rank.
If it does not have full rank, this corresponds to even more deletions than given by $\rho$.
Thus, if we work on a row erasure channel with no deletions, we may assume that~$A$ has full rank~$k$.\\
Instead of deleting the rows with erasures in them, we can also fill the respective rows with ?'s. Then we can equivalently model the output of the channel as
$$\hat{R}=AU+ E\mathbf 1_n ,$$
where $\mathbf 1_n\in\F_q^{n\times n}$ is the matrix with all entries equal to~$1$.
Since  a symbol erasure leads to erasing or disregarding the entire affected vector at the receiving node, we call this channel model the \emph{row erasure channel (REC) model}.

\medskip

The second model, dealing with symbol erasures (and more generally symbol errors), has been introduced by Skachek/Milenkovic/Nedi{\'c}~\cite{sk13}.
Suppose that a symbol erasure appears in the $i$th entry of a certain vector, say~$\mathbf v$.
In this model the node does not delete the affected vector, but rather transmits it as usual, using the identities in~\eqref{e-QuestMark}. Thus, at the receiver side all vectors that were produced as linear
combinations involving~$\mathbf v$ have an erasure in the $i$th entry.
This is regardless of where in the network the erasure occurred, which justifies to assume the worst case that erasures occur at the source.
Randomness of the network then requires us to assume that all received vectors have an erased $i$th entry.\footnote{
Our use of the terminology `symbol erasure' differs from the use in~\cite{sk13}. In the latter it is used for describing an erased entry for all vectors obtained by the receiver. 
We will call this a `column erasure'.}

\medskip

\textbf{Column Erasure Channel Model.}
Define the \emph{column erasure operator} $\gamma$ on the matrix space $(\F_q\cup\{?\})^{k\times n}$ to replace every column of the matrix that contains at least one symbol erasure with an all-erasure column.
The channel takes as input a matrix $U \in \F_q^{k\times n}$ and outputs
\begin{equation}\label{e-SMN}
  \tilde{R} = \gamma(A U + E) \quad \in \left(\F_q\cup\{?\}\right)^{k\times n},
\end{equation}
where $A\in \F_q^{k\times k}$ is the representation of the random operations of the inner nodes of the network channel and $E\in \{0, ?\}^{k\times n}$ is the symbol erasure matrix.
We can equivalently write
$$\tilde{R}=AU+\mathbf 1_k E ,$$
where $\mathbf 1_k\in\F_q^{k\times k}$ is the matrix with all entries equal to~$1$.
Note that this model does not distinguish between symbol erasures occurring in the same transmitted vector and those in different vectors.
Since a single symbol erasure at the $i$th entry results in a completely erased $i$th column of the received matrix, we call this model the \emph{column erasure channel (CEC) model}.

\medskip

The above channel models represent what the receiver sees. The effect of a single symbol erasure at some edge in the network can be quite different.
For instance, an affected vector at some inner node does not have any implications if it is not transmitted further, e.g., if the respective scalar of the linear combination at the inner node is zero.
However, because of the randomness of the network we cannot distinguish such cases and thus have to assume that every erasure will propagate as much as possible through the network. This assumption was also done in \cite{sk13} and more explanations on this assumption can be found in there.
Therefore, we have the following worst case scenarios:

\begin{lem}\label{L-WCSell}
Suppose~$\ell$ symbol erasures happened (that is, $\ell$ entries of $E$ are a $?$).  In the worst case we have
\begin{enumerate}
\item $\ell$ row deletions in the received matrix in the REC.
      For this to happen the erasures have to appear in~$\ell$ different rows of~$E$.
\item $\ell$ column erasures in the received matrix in the CEC.
      For this to happen the erasures have to appear in~$\ell$ different columns of~$E$.
\end{enumerate}
\end{lem}

Note again that the location of the erasures in the matrix~$E$ is only a necessary condition for the worst case because
in a specific instance of the network the erasure may not occur in any linear combination that is transmitted downstream.
However, throughout the paper, we assume the worst case where symbol erasures affect a maximum number of vectors.

\begin{rem}
In the following three sections we assume that the channel operation matrix $A$ has full rank, i.e., that no deletions have occurred during the transmission. The case with deletions will be handled in Section \ref{sec:simul}. From an application point of view deletion-free transmission can be achieved e.g. by using a fountain mode, as explained in \cite{si17}, or simply by declaring a failure when the received space has lower dimension than required. In the latter case the probability that $A$ is rank deficient tends to zero with growing field size $q$ or dimension $k$.
\end{rem}


\section{Spread Codes and Symbol Erasures}\label{sec:spreads}

In this section we investigate the performance of spread codes in $\G$ in the row erasure channel and the column erasure channel, assuming that no deletions occurred.
For this we first make a worst case analysis and then a more detailed analysis, counting the exact number of erasure matrices $E\in \{0,?\}^{k\times n}$ (which we call \emph{erasure patterns})  that can be decoded by the receiver for any random linear combinations taken at the inner nodes. For simplicity we include the zero matrix in this count, although this technically corresponds to no erasures at all.


\subsection{Spread Codes in the Row Erasure Channel (REC)}\label{S-SpreadKK}
We first investigate the capability of spread codes with respect to symbol erasure decoding in the row erasure channel model.

\begin{thm}\label{thm:classKK}
Let $\C\subseteq\G$ be a spread code.
In the REC, the code~$\C$ can correct any erasure pattern $E\in \{0,?\}^{k\times n}$ with at most $k-1$ nonzero entries.
On the other hand, there exist erasure patterns in $\{0,?\}^{k\times n}$ with $k$ nonzero entries that cannot be corrected.
Thus, the symbol erasure correction capability in the classical sense is $k-1$  for the REC.
\end{thm}

\begin{proof}
From~\cite[Thm.~2]{ko08} we know that in the REC-model, the code~$\C$ can correct $k-1$ deletions (i.e., dimension losses).
Lemma~\ref{L-WCSell} shows that in the worst case~$k-1$ symbol erasures lead to $k-1$ row erasures.
Similarly, Lemma~\ref{L-WCSell} implies that in the worst case $k$ symbol erasures lead to $k$ row erasures. The resulting empty matrix cannot be decoded.
\end{proof}

However, if we consider all possible erasure patterns, there are a lot more that we can actually correct. Since erasure patterns are represented by the erasure matrices $E\in \{0,?\}^{k\times n}$, we will count the number of these matrices that are correctable at the receiver side.

\begin{thm}\label{T-KKcount}
Considering the REC, there are
\[
     2^{kn}  - (2^n-1)^k
\]
symbol erasure patterns $E\in \{0,?\}^{k\times n}$ that can be corrected by  a spread code $\C\subseteq \G$.
\end{thm}

\begin{proof}
Overall we have $ 2^{kn}$ possible erasure patterns.
Since $k-1$ dimension losses can be corrected, the only erasure patterns we cannot correct are the ones that
have a~$?$ in each of the~$k$ rows.
There are $2^n-1$ possibilities for a nonzero row in $\{0,?\}^n$, which results in $(2^n-1)^k$ non-correctable patterns.
As a consequence there are $ 2^{kn} - (2^n-1)^k$ correctable erasure patterns.
\end{proof}


\subsection{Spread Codes in the Column Erasure Channel (CEC)}\label{S-SpreadSMN}

In this section we consider the Desarguesian spread codes presented in Definition~\ref{D-DesSpread} for the CEC.
The following is the analog of Theorem \ref{thm:classKK} for the column erasure channel model.

\begin{thm}\label{thm:classSMN}
Let $n=mk$ and $\C=\mathcal S_q(m,k,P)\subseteq \G$ be a Desarguesian spread code.
On the CEC, the code~$\C$ can correct any erasure pattern $E\in \{0,?\}^{k\times n}$ with at most $k-1$ nonzero columns.
On the other hand, there exist erasure patterns in $\{0,?\}^{k\times n}$ with~$k$ nonzero entries that cannot be corrected.
Thus, the symbol erasure correction capability in the classical sense is $k-1$ for the CEC.
\end{thm}

\begin{proof}
Let $\Uvs \in \C$ be a codeword and $U=(U_1 \mid \ldots\mid U_m)$ with $U_i \in \F_q[P]$ its matrix representation.
Let $R=\gamma(AU+E) = (R_1 \mid \ldots\mid R_m)$ be the received matrix, as in \eqref{e-SMN}.
Since we do not consider any deletions,~$A$ is invertible. Therefore, $AU_i$ is either invertible or zero, for $i=1,\dots,m$.
Hence, as a first step, we can decode any block with at least one zero column to a zero block.

Let $R_\ell$ be any block of $R$ that does not contain a zero column.
Then $R_\ell$ coincides with~$A U_\ell$ in the non-erased columns.
Let $H\subseteq\GL_k(q)$ be the set of all invertible matrices that coincide with~$R_\ell$ in the non-erased columns.
Then $H\neq\emptyset$ because $A U_\ell\in H$.
For any $Z\in H$ consider the matrix $Z^{-1}R$.
For $Z=A U_\ell$, the matrix $Z^{-1}R$ agrees with~$U_\ell^{-1}U$ in the non-erased columns.
Hence all blocks of  $Z^{-1}R$ are partially erased matrices from the rank-metric code~$\F_q[P]$  and can thus be decoded using Lemma~\ref{L-RP}.

It remains to show that for any other choice of $Z\in\GL_k(q)$ for which every block of $Z^{-1}R$ can be decoded in
$\F_q[P]$, the decoding leads to the same subspace in~$\C$.
To this end we may assume without loss of generality that $Z_1,Z_2\in\GL_k(q)$ are such that
\[
    Z_1^{-1} R = \gamma\big( ( I \mid X_2 \mid \dots \mid X_m) + E \big) \ \text{ and }\
   Z_2^{-1} R = \gamma\big( ( I \mid Y_2 \mid \dots \mid Y_m) + E \big) ,
\]
where all blocks are decodable  in the rank-metric code $\F_q[P]$.
Decode every block in $\F_q[P]$ and denote the solutions by $X,Y \in \F_q^{k\times n}$, respectively.
Hence, $\rs(X),\rs(Y) \in \C$. We know that $Z_1 X$ and $Z_2 Y$ agree on at least $n-(k-1)$ columns (the non-erased columns of $R$).
Thus, there exists $Q\in \GL_n(q)$ such that
\[
   Z_1XQ = (M\mid A)\ \text{ and }\
   Z_2YQ = (M\mid B),
\]
where $M\in \F_q^{k\times (n-t)}$ and $A,\,B\in \F_q^{k\times t}$ with $t\leq k-1$.
We obtain
\begin{align*}
    d_S(\rs(X), \rs(Y) ) &= d_S (\rs(Z_1 X Q), \rs(Z_2 Y Q))\\
      &= 2\rk\left( \begin{array}{cc}M & A \\ M &B \end{array}\right) -2k = 2(\rk(M) +\rk (A-B) -k) \leq 2k-2 .
\end{align*}
Since $\rs(X),\rs(Y) $ are both codewords of the spread~$\C$, we conclude that they must be equal.
All of this shows that the above described decoding is unique.

For the second statement note that by Lemma~\ref{L-WCSell} in the worst case~$k$ symbol erasures lead to~$k$ column erasures.
If these $k$ column erasures occur in one block, we cannot recover the codeword.
\end{proof}

Thus, the classical symbol erasure correction capability of spread codes is the same in both the REC and CEC.
However, the actual number of correctable erasure patterns is different, as we show in the following.


\begin{thm}\label{lem:k-1}
Let $n=mk$ and $\C = \mathcal S_q(m,k,P)\subseteq \G$ be a Desarguesian spread code.
In the CEC, any column erasure pattern $E\in\{0,?\}^{k\times n}$, for which the matrix of the sent codeword has at most
$k-1$ columns per block affected by erasures and one nonzero block is unaffected by erasures,
can be uniquely decoded.
\end{thm}

\begin{proof}
Let $\Uvs \in \C$ be a codeword and $U=(U_1 \mid \ldots\mid U_m)$ with $U_i \in \F_q[P]$ its matrix representation.
As before let $R=\gamma(AU+E) = (R_1 \mid \ldots\mid R_m)$ be the received matrix for some $A\in\GL_k(q)$.
Without loss of generality let the first block be nonzero and unaffected by column erasures.
Considering only the first block and one more block, say the $i$th one, we arrive at the situation of Theorem \ref{thm:classSMN} for the spread code $\mathcal S_q(2,k,P)$.
Thus we can uniquely recover $\rs(U_1\mid U_i)= \rs(I_k\mid U_1^{-1}U_i)$ from $(R_1\mid R_i)$ for any $i=2,\dots,m$. This results in the unique codeword $ \rs(I_k\mid U_1^{-1}U_2 \mid \dots \mid U_1^{-1}U_m) = \Uvs$.
%
\end{proof}

The following example shows that the assumption of one unaffected block is necessary for decodability.
\begin{ex}\label{E-Blocks}
In $\cG_2(3,6)$ consider the spread $\mathcal S_2(2,3,P)$, where
\[
    P=\begin{pmatrix}0&0&1\\1&0&0\\0&1&1\end{pmatrix}.
\]
Furthermore, consider the invertible matrices
\[
   A_1=\begin{pmatrix}1&1&1\\0&1&1\\1&1&0\end{pmatrix}, \
   A_2=\begin{pmatrix}1&1&0\\0&1&1\\1&0&0\end{pmatrix}.
\]
Then the matrices
\[
   A_1(I_3\mid P^3)=\left(\!\!\begin{array}{ccc|ccc}1&1&1&0&1&0\\0&1&1&1&0&1\\1&1&0&1&0&0\end{array}\!\!\right)\
   \text{ and }\
   A_2(I_3\mid P^5)=\left(\!\!\begin{array}{ccc|ccc}1&1&0&0&1&1\\0&1&1&1&0&0\\1&0&0&1&0&1\end{array}\!\!\right)
\]
represent different codewords in~$\mathcal S_2(2,3,P)$.
After erasing the last two columns of the first block and the last column of the second block, the resulting matrices are not distinguishable anymore.
This shows that this pattern of at most~$k-1$ column erasures per block is not decodable.
However, if the last block has no erasures we can uniquely reconstruct both codewords.
\end{ex}

For the CEC, the number of correctable erasure patterns depends on the transmitted codeword.
The precise version is as follows (see also Remark~\ref{R-count} after the proof).

\begin{thm}\label{T-SMNcount}
Consider a Desarguesian spread code $\mathcal S_q(m,k,P)\subseteq \G$.
Suppose the row space of the matrix
\[
  U=(0\mid\ldots\mid0\mid I\mid B_{i+1}\mid \ldots\mid B_m)\in\cM
\]
is transmitted over the CEC.
Let~$\ell$ be the number of nonzero blocks~$B_j$ and $N:=2^{k^2} - (2^k-1)^k$.
Then at least $e_\ell:=N^m(1-(\frac{N-1}{N})^{\ell+1})$ symbol erasure patterns $E\in \{0,?\}^{k\times n}$ can be uniquely decoded.
As a consequence, the code $\mathcal S_q(m,k,P)$ can correct on average (at least)
\[
   \eavg:=\frac{N^{m}}{q^n-1}\Big(q^n-\Big[\frac{(q^k-1)(N-1)}{N}+1\Big]^m\Big)
\]
symbol erasure patterns in the CEC.
\end{thm}

\begin{proof}
In order to normalize the received word we need at least one nonzero block received correctly.
That gives us $\ell+1$ choices.
In each of the remaining $m-\ell-1$ blocks we can correct up to $k-1$ column erasures.
As in the proof of Theorem~\ref{T-KKcount} this yields~$N$ correctable symbol erasure patterns (including $E=0$) per block.
Denoting by~$t$ the number of correct nonzero blocks, we have $\ell+1-t$ blocks with at least one erasure, and thus we
obtain
\[
     e_\ell:=\sum_{t=1}^{\ell+1} \binom{\ell+1}{t}(N-1)^{\ell+1-t}N^{m-\ell-1}
\]
possibilities of correctable symbol erasure patterns for all blocks combined (including the zero blocks).
Via the term for $t=\ell+1$ this count includes the erasure pattern $E=0_{k\times n}$.
The above simplifies to
\begin{align}
  e_\ell&=N^{m-\ell-1}\Big(\sum_{t=0}^{\ell+1} \binom{\ell+1}{t}(N-1)^{\ell+1-t}-(N-1)^{\ell+1}\Big)\nonumber\\
     &=N^{m-\ell-1}(N^{\ell+1}-(N-1)^{\ell+1})=N^m\Big(1-\big(\frac{N-1}{N}\big)^{\ell+1}\Big),\label{e-el}
\end{align}
as stated.
For the second statement note
that there are $\binom{m-i}{\ell}(q^k-1)^{\ell}$ matrices in~$\cM$ of the form~$U$ as given in the theorem.
Hence the average number of correctable symbol erasure patterns is
\begin{align}
\eavg&=\frac{\sum_{i=1}^m\sum_{\ell=0}^{m-i}\binom{m-i}{\ell}(q^k-1)^\ell e_\ell}{(q^n-1)/(q^k-1)}
   =\frac{1}{q^n-1}\sum_{\ell=0}^{m-1}(q^k-1)^{\ell+1}e_\ell\sum_{i=1}^{m-\ell} \binom{m-i}{\ell}\nonumber\\
   &=\frac{1}{q^n-1}\sum_{\ell=0}^{m-1}\binom{m}{\ell+1}(q^k-1)^{\ell+1}e_\ell.\label{e-eavg}
\end{align}
Using the expression for~$e_\ell$ we further derive
\begin{align*}
\eavg&=\frac{N^m}{q^n-1}\Big[\sum_{\ell=0}^{m-1}\binom{m}{\ell\!+\!1}(q^k-1)^{\ell+1}
   -\sum_{\ell=0}^{m-1}\binom{m}{\ell\!+\!1}\Big(\frac{(q^k-1)(N-1)}{N}\Big)^{\ell+1}\Big]\\
   &=\frac{N^m}{q^n-1}\Big[\sum_{\ell=0}^{m}\binom{m}{\ell}(q^k-1)^{\ell}
   -\sum_{\ell=0}^{m}\binom{m}{\ell}\Big(\frac{(q^k-1)(N-1)}{N}\Big)^{\ell}\Big]\\
   &=\frac{N^m}{q^n-1}\Big[q^n-\Big(\frac{(q^k-1)(N-1)}{N}+1\Big)^m\Big],
\end{align*}
where the last step follows from $n=mk$.
\end{proof}

\begin{rem}\label{R-count}
For simplicity we only counted the erasure patterns discussed in Theorem~\ref{lem:k-1}.
If $m<k$ the number of correctable erasure pattern is even higher according to Theorem~\ref{thm:classSMN}, since we can also correct erasure
patterns affecting each of the $m$ blocks, as long as at most $k-1$ columns are erased.
\end{rem}

The following remark depicts another scenario where more erasure patterns than stated in Theorem \ref{T-SMNcount} can be corrected.

\begin{rem}
Note that $e_0=N^{m-1}$ is the number of correctable symbol erasure patterns in the case that the only nonzero block in the
matrix~$U$ is the identity matrix (regardless of its position).
In this case, we can in fact correct more symbol erasures.
Indeed,~$N^{m-1}$ gives us the number of symbol erasures that let us recover the $m-1$ zero blocks.
But that information is already sufficient to conclude that the remaining block has to be the identity matrix.
In other words, we can tolerate up to~${k^2}$ symbol erasures in that block.
Hence for these particular matrices the number of correctable symbol erasure patterns is $N^{m-1}2^{k^2}$.
\end{rem}

The quite complicated formula for $\eavg$ from Theorem \ref{T-SMNcount} can asymptotically be simplified to $mN^{m-1}$, as shown in the following.
Intuitively, $mN^{m-1}$ can be interpreted as the number of all erasure patterns $E=(E_1\mid\ldots\mid E_m)\in\{0,?\}^{k\times n}$, where one block of~$E$ is completely zero and all other blocks have at least one complete zero column.

\begin{prop}\label{P-eavgLimit}
Recall that $N:=2^{k^2} - (2^k-1)^k$. Fix~$m\in\N$ and let $n=mk$. Then
\[
  \lim_{k\to\infty}\,\frac{\eavg}{mN^{m-1}}=1.
\]
\end{prop}
\begin{proof}
We compute
\begin{align*}
\eavg&=\frac{1}{q^n-1}\Big[N^mq^n-\sum_{i=0}^m\binom{m}{i}(q^k-1)^i(N-1)^iN^{m-i}\Big]\\
     &=\frac{1}{q^n-1}\Big[N^mq^n-\sum_{i=0}^m\binom{m}{i}(q^k-1)^i\sum_{j=0}^i\binom{i}{j}(-1)^{i-j}N^{m+j-i}\Big]\\
     &=\frac{1}{q^n-1}\Big[N^m\underbrace{\Big(q^n-\sum_{i=0}^m\binom{m}{i}(q^k-1)^i\Big)}_{=0}
        -\sum_{i=0}^m\binom{m}{i}(q^k-1)^i\sum_{j=0}^{i-1}\binom{i}{j}(-1)^{i-j}N^{m+j-i}\Big]\\
     &=-\sum_{i=0}^m\binom{m}{i}\frac{(q^k-1)^i}{q^{mk}-1}\sum_{j=0}^{i-1}\binom{i}{j}(-1)^{i-j}N^{m+j-i}.
\end{align*}
Hence
\[
  \frac{\eavg}{mN^{m-1}}=\frac{-1}{m}\sum_{i=0}^m\binom{m}{i}\frac{(q^k-1)^i}{q^{mk}-1}\sum_{j=0}^{i-1}\binom{i}{j}(-1)^{i-j}N^{1+j-i}.
\]
Since $\lim_{k\to\infty}\frac{(q^k-1)^i}{q^{mk}-1}=0$ for $i<m$ and $\lim_{k\to\infty}N^{1+j-i}=0$ for $j<i-1$ we conclude
\[
  \lim_{k\to\infty}\,\frac{\eavg}{mN^{m-1}}=\frac{-1}{m}\binom{m}{m}\frac{(q^k-1)^m}{q^{mk}-1}\binom{m}{m-1}(-1)^{1}N^{0}=1.
  \qedhere
\]
\end{proof}

\subsection{Comparison of Symbol Erasure Correction Capabilities}\label{S-CompSpread}

In this section we compare the symbol erasure correction capabilities of the REC and CEC for spread codes.
As shown in Theorems \ref{thm:classKK} and \ref{thm:classSMN} the classical symbol erasure correction capability is equal for both channel models, namely $k-1$.

However, as we show next, the number of decodable erasure patterns in the CEC exceeds
the number of decodable erasure patterns in the REC by an exponential factor.

\begin{thm}
Let $n=mk$ and consider a Desarguesian spread code  $\mathcal S_q(m,k,P)\subseteq \G$.
Denote the number of correctable erasure patterns for the REC by $r(n,k)=2^{kn}- (2^n-1)^k$
(see Theorem~\ref{T-KKcount}) and the average number of correctable erasure patterns for the CEC by
$\eavg(n,k)$ as given in Theorem~\ref{T-SMNcount}.
Then
\[
   \frac{r(n,k)}{\eavg(n,k)}\leq kN\Big(\frac{2^{(k-1)k}}{N}\Big)^m.
\]
As a consequence, for any fixed~$k$ we have
$\lim_{m\to\infty}\frac{r(n,k)}{\eavg(n,k)}=0$.
\end{thm}

\begin{proof}
First of all, it is straightforward to show that for all $k<n$
\[
    2^{(k-1)n}< r(n,k)< k2^{(k-1)n}.
\]
Thus $N:=2^{k^2}-(2^k-1)^k$ satisfies $N>2^{(k-1)k}$.
Next, by~\eqref{e-el} we have $e_\ell=N^m(1-(\frac{N-1}{N})^{\ell+1})$, and from this one easily derives
$e_\ell\geq N^{m-1}$.
As a consequence, using~\eqref{e-eavg} we obtain
\[
  \eavg(n,k)\geq N^{m-1}\frac{1}{q^n-1}\sum_{\ell=1}^{m}\binom{m}{\ell}(q^k-1)^{\ell}=
  N^{m-1}\frac{1}{q^n-1}\big((q^k-1+1)^m-1\big)
  = N^{m-1},
\]
where the last step follows from $mk=n$.
Thus
\[
  \frac{r(n,k)}{\eavg(n,k)}\leq \frac{k2^{(k-1)n}}{N^{m-1}}=kN\Big(\frac{2^{(k-1)k}}{N}\Big)^m,
\]
as stated.
Now the limit follows from the fact that for fixed~$k$ the fraction $\frac{2^{(k-1)k}}{N}$ is a constant strictly less than~$1$.
\end{proof}

The following figures depict the number of correctable erasure patterns from Theorems~\ref{T-KKcount} and \ref{T-SMNcount} for $q=2$ and $k=3,4$. Recall that $n=mk$.
The graphs show that only for very small~$n$ the row erasure model outperforms the column erasure model;
for growing~$n$ the column erasure model is preferable.

\begin{center}
\includegraphics[width=8cm]{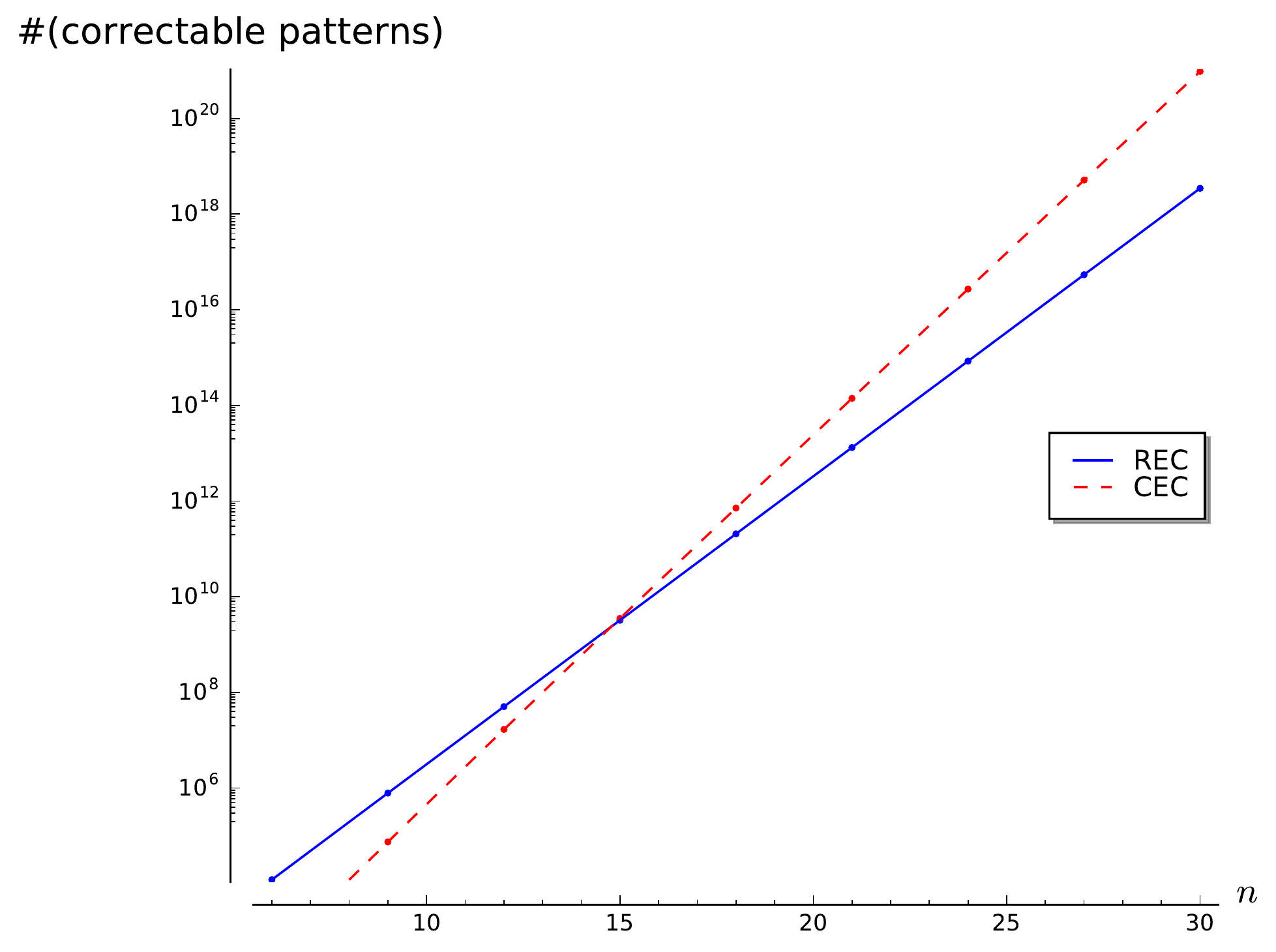}\\
 $k=3$
\end{center}

\begin{center}
\includegraphics[width=8cm]{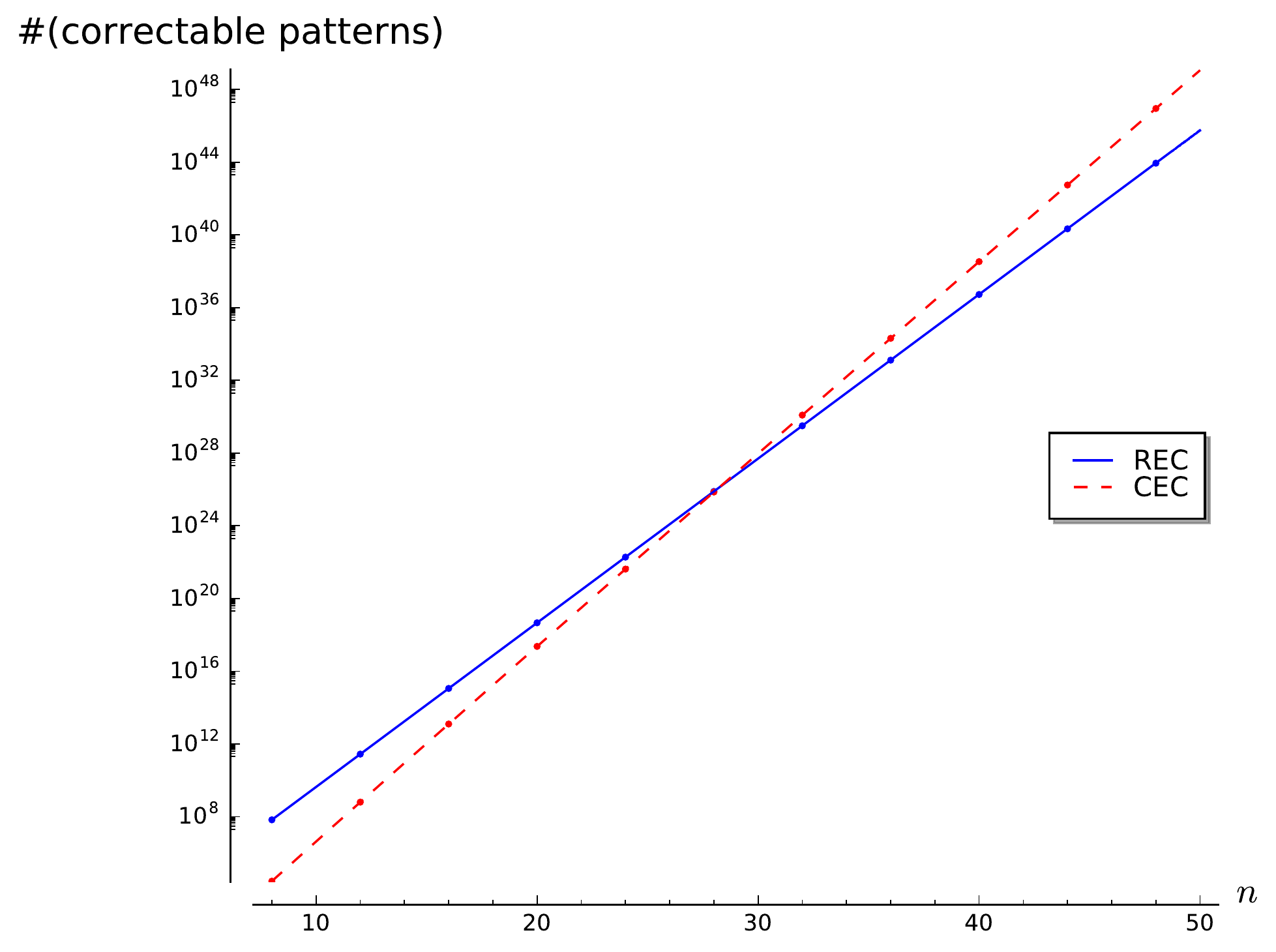}\\
$k=4$
\end{center}
%
%


\section{Comparison to Hybrid Codes}\label{S-Hybrid}

In this section we briefly describe the hybrid codes developed by Skachek et al.~\cite[Sec. V]{sk13} and compare their erasure correction capability
on the CEC to that of spread codes.
The channel model studied in~\cite{sk13} is exactly the CEC introduced in Section~\ref{S-Models}, i.e.,
 a column erasure at position~$i$ is defined as an erased~$i$th entry for \emph{all} received vectors.
The authors suggest to use an $[n,n']_q$-generalized Reed-Solomon (GRS) code interleaved with a subspace code in $\mathcal G_q(k,n')$ to obtain a
good column erasure (and dimension-error) decoding performance.
This implies that the field size must satisfy\footnote{Their construction of hybrid codes can easily be generalized to use extended Reed-Solomon codes instead of GRS codes, which would increase the lower bound on the field size by one. However, this does not make much of a difference for our analysis, therefore we use the original construction with GRS codes.} $q\geq n$.

We note that hybrid codes are designed specifically for the use in the CEC. In the REC the inner Reed-Solomon code would have no purpose, which is why we only consider hybrid codes over the CEC.

\begin{defi}\label{D-Hybrid}
Let $\C\subseteq \mathcal G_q(k,n')$ be a constant-dimension code. 
Furthermore, let $C_{\text{RS}}=\rowsp(G)$ be an $[n,n']_q$-GRS code with generator matrix $G\in\F_q^{n'\times n}$.
Then the subspace code
\begin{equation}\label{e-hyb}
  \Big\{\{\mathbf v G\mid \mathbf v\in \U \}\,\Big|\, \U\in\C\Big\}\subseteq\G,
\end{equation}
is called a $[n,k,n']_q$-\emph{hybrid code}.
\end{defi}

\begin{lem}\cite[Thm.~V.1]{sk13}\label{L-HybEras}
Let~$\C\subseteq \mathcal G_q(k,n')$ be a constant-dimension code with $d_S(\C)=2D$ and $\cH$ be the corresponding $[n,k,n']_q$-hybrid code of the form \eqref{e-hyb}.
Then~$\cH$ can correct up to $D-1$ dimension errors and $n-n'$ column erasures.
\end{lem}

As a consequence, if we only want to deal with symbol erasures, but no dimension errors, we may choose $D=1$.
In this case the constant-dimension code~$\C$ is the entire Grassmannian $\cG_q(k,n')$ and thus
\begin{equation}\label{e-L}
  \cH:=\Big\{\{\mathbf v G\mid \mathbf v\in \U\}\,\Big|\, \U\in\cG_q(k,n')\Big\}.
\end{equation}
Hence this code can correct up to $n-n'$ column erasures and its cardinality is $\Gaussian{n'}{k}_q$.
Note that we have freedom in choosing~$n'\in\{k+1,\ldots,n-1\}$ in order to optimize the performance of~$\cH$.
The number of correctable column erasures translates straightforwardly into the following.

\begin{lem}\label{L-HybSymbEras}
The number of correctable symbol erasure patterns $E\in\{0,?\}^{k\times n}$ (including the zero matrix) for the hybrid code~$\cH$ from~\eqref{e-L} is
\[
   e_{\mathcal H}:=\sum_{j=0}^{n-n'} \binom{n}{j} (2^k-1)^j.
\]
\end{lem}

In order to compare the performance of a spread code and a hybrid code on the CEC we need to take the rate into account.
Recall from~\cite[Def.~IV.7]{sk13} that the rate of a subspace code~$\C$ in~$\G$ is
$\log_q(|\C|)/(nk)$.
Thus the rate of a spread code~$\cS$ in $\G$, where $n=km$, is
\begin{equation}\label{e-Rspread}
   \Rs=\frac{\log_q(|\cS|)}{{nk}} = \frac{1}{nk}\Big(\log_q\frac{q^n-1}{q^k-1}\Big) \approx \frac{n-k}{nk}=\frac{mk-k}{mk^2}\approx\frac{1}{k}.
\end{equation}
On the other hand, the rate of the hybrid code~$\cH$ as in~\eqref{e-L} is
\begin{equation}\label{e-Rhybrid}
   \Rh=\frac{\log_q(|\cH|)}{{nk}} =\frac{\log_q(|\mathcal G_q(k,n') |)}{nk} = \frac{\log_q\big(\Gaussian{n'}{k}_q\big)}{nk}\approx
   \frac{n'-k}{n},
\end{equation}
where the last approximation follows from
\[
  \log_q\big(\Gaussian{n'}{k}_q\big)=\log_q\Big(\prod_{i=0}^{k-1}\frac{q^{n'-i}-1}{q^{k-i}-1}\Big)
  \approx\log_q\Big(\prod_{i=0}^{k-1}q^{n'-k}\Big)=k(n'-k).
\]


For comparability let us now fix the same dimension~$k$ and length~$n$ for both codes and find~$n'$ so that the hybrid code and the spread have approximately the same rate.
For small~$k$, the rate of the spread code is approximately $(n-k)/nk$ by~\eqref{e-Rspread}, and thus~\eqref{e-Rhybrid} tells us that
we need $(n'-k)/k=(n-k)/nk$. This in turn is equivalent to
\[
    n'=(n-k)/k+k.
\]
Hence by Lemma~\ref{L-HybEras} the hybrid code can correct at most $n-n'=n-n/k-k+1$ column erasures.
From Theorem~\ref{lem:k-1} we know that the spread code $\cS\subseteq\G$ can correct at most $(n/k-1)(k-1)$ erased columns, which is also $=n-n/k-k+1$.
However, the hybrid code can correct any combination of those columns, whereas the spread can only correct certain combinations of the columns (and in some codewords even less).
Especially for large~$n$, compared to~$k$, this works in favor of hybrid codes.

\begin{ex}\label{ex19}
We fix $k=2$ and $n'=n/2+1$ for variable $n$.
Moreover, for the spread codes we fix $q=2$, whereas for the hybrid codes we pick $q$ as the smallest prime power exceeding $n-1$.
We obtain the following data.
$$
\begin{array}{|c||c|c||c|c|}
\hline
   n \phantom{\Big|}&
   \textnormal{rate spread} & \textnormal{rate hybrid} & \eavg \text{ (see Thm.~\ref{T-SMNcount})}& e_{\mathcal H}\text{ (see Lem.~\ref{L-HybSymbEras})}\\\hline
6 \phantom{\Big|}& 0.366 & 0.341 & 100 & 154\\\hline
8 \phantom{\Big|}& 0.401 & 0.379 & 879 & 1789\\\hline
10 \phantom{\Big|}& 0.421 & 0.402 & 7277 & 20686\\\hline
12 \phantom{\Big|}& 0.434 & 0.418 & 58059 &239122\\\hline
14 \phantom{\Big|}& 0.443 & 0.429 & 451041 & 2767444\\\hline 
  \end{array}
$$
One can see that the rate of the spread code is slightly higher while the number of correctable erasure patterns is less compared to the hybrid code. If we increase $n'$ by one, i.e., $n'=n/2+2$, we get the following data.
$$
\begin{array}{|c||c|c||c|c|}
\hline
   n \phantom{\Big|}&
   \textnormal{rate spread} & \textnormal{rate hybrid} & \eavg \text{ (see Thm.~\ref{T-SMNcount})}& e_{\mathcal H}\text{ (see Lem.~\ref{L-HybSymbEras})}\\\hline
6 \phantom{\Big|}& 0.366 & 0.507 & 100 & 19\\\hline
8 \phantom{\Big|}& 0.401 & 0.504 & 879 & 277\\\hline
10 \phantom{\Big|}& 0.421 & 0.502 & 7277 & 3676\\\hline
12 \phantom{\Big|}& 0.434 & 0.501 & 58059 & 46666\\\hline
14 \phantom{\Big|}& 0.443 & 0.501 & 451041 & 578257\\\hline
  \end{array}
$$
In this case the hybrid code has larger rate, but for small $n$ the spread code can correct more erasure patterns.

\end{ex}

\begin{rem}
If we fix $k,n$ to be the same for both codes, then for large~$n$ hybrid codes outperform spread codes with respect to rate and erasure correction capability. However, for small~$n$, there are parameter sets where spread codes have a better rate or better erasure correction capability than hybrid codes. Moreover, spread codes have the immense advantage that they exist over any field,
whereas hybrid codes need a field size $q\geq n$.
\end{rem}


We can also compare spread and hybrid codes without assuming that $k$ and $n$ are the same for both codes. For comparability we fix the rate of the codes to be approximately the same and compare their erasure correction capability, as shown in the next example.

\begin{ex}\label{E-HybSpr}
We start with a hybrid code~$\cH$ in $\cG_{29}(10,25)$, hence $k=10$ and $n=25$ (recall that the field size has to be at least~$n$).
The rate and performance depend on the choice of $n'\in\{k+1,\ldots,n-1\}$.
Let us pick $n'=13$.
Then~\eqref{e-Rhybrid} and Lemma~\ref{L-HybSymbEras} lead to the rate and erasure pattern correction capability
\[
    \Rh=0.12004, \quad \eh=0.68\cdot10^{43},
\]
respectively.
We want to find a spread code over~$\F_{29}$ with approximately the same rate and compare its erasure pattern correction capability with the one of the hybrid code.
Denote the length and dimension of the spread code by~$\tilde{n},\,\tilde{k}$, respectively.
Since the erasure pattern correction capability is the number of correctable erasure matrices $E\in\{0,?\}^{\tilde{k}\times\tilde{n}}$,
a fair comparison should consider the proportion of correctable erasure matrices.
Thus we aim for a spread code with rate $\Rs\approx\Rh=0.12004$ and then want to compare its proportion of correctable erasure patterns to the proportion
\[
   \frac{\eh}{2^{kn}}=10^{-33}.
\]
As for the parameters of the spread code,~\eqref{e-Rspread} and~\eqref{e-Rhybrid} show that $\tilde{k}\approx n/(n'-k)=25/3$.
Let us consider the interval $\{6,\ldots,10\}$ about this value.
For each value of $\tilde{k}$ in $\{6,\ldots,10\}$, we then find the smallest $\tilde{n}=m\tilde{k}$ such that the resulting rate~$\Rs$ is
larger than~$\Rh$. For that code we list the normalized  erasure pattern correction capability $\eavg/2^{\tilde{k}\tilde{n}}$, where~$\eavg$ is
average number of correctable erasure patterns as in Theorem~\ref{T-SMNcount}.
This leads to the following table (for $\tilde{k}>8$ the spread code never has a larger rate).
\[
  \begin{array}{|c|c|c|c|}
    \hline
    [\tilde{k},\tilde{n}]\phantom{\Big|}&\text{Rate $\Rs$} &\eavg/2^{\tilde{k}\tilde{n}}\\ \hline\hline
    [6,24]\phantom{\Big|}& 0.12500 & 10^{-14}\\ \hline
    [7,49]\phantom{\Big|}& 0.12245  & 10^{-22}\\ \hline
    [8,208]\phantom{\Big|}& 0.12019 & 10^{-56}\\ \hline
  \end{array}
\]
By design, in all cases the spread code has a slightly larger rate than the hybrid code~$\cH$.
We observe that for $\tilde{k}=6,7$ the spread code can correct a much larger proportion of erasure patterns than the hybrid code, whereas
for $\tilde{k}=8$ the hybrid code can correct a larger proportion.

In the same way we can choose other values for $n'\in\{k+1,\ldots,n-1\}$.
It turns out that for $n'<13$, the hybrid code is always better (in terms of the proportion of correctable erasure patterns), whereas for
$n'>13$ the spread code is better.
\end{ex}


\section{Decoding Complexities}\label{sec:decoding}

In this section we describe two decoding algorithms for Desarguesian spread codes, one in the row erasure channel and one in the column erasure channel model. We derive their complexity orders and compare them to the complexity of the decoder for hybrid codes from \cite{sk13}, which is based on a decoder for Reed-Solomon codes.

We will make use of the maps~$\phi,\,\psi$, and $\bar\psi$ from~\eqref{e-phi} -- \eqref{e-barpsi}.
Moreover, recall the isomorphism~\eqref{e-DesGrass} between a Desarguesian spread and a Grassmannian.
For decoding a Desarguesian spread $\mathcal S_q(m,k,P)$ or $\mathcal S_q(m,k,P^\top)$  it thus suffices to recover the isomorphic representation of a codeword in $\cG_{q^k} (1,m)$. For more information on message encoding for Desarguesian spread codes see \cite{ho17}.

But even if one wishes to recover the original codeword in $\G$, the complexity of finding a representation in the original~$\mathcal S_q(m,k,P)$ is as follows.

\begin{prop}\label{P-DecGrassDes}
One can obtain a basis matrix of the original codeword in $\mathcal S_q(m,k,P)$ from its representation  in
$\cG_{q^k} (1,m)$ with $O(k^2 m)=O(kn)$ operations over~$\F_q$.
\end{prop}
\begin{proof}
In~\cite[Lemma 17]{ho17} it is shown that the map~$\phi$ can be carried out in $O(k^2)$ operations over the field~$\F_q$. Since this needs to be done for any of the $m$ blocks
of the codeword matrix representation, the statement follows.
\end{proof}

From now on we focus on recovering the spread codeword as an element in $\cG_{q^k} (1,m)$.
In  order to have a unique representation we will always recover the normalized basis vector $(u_1,\ldots,u_m)\in\F_{q^k}^m$, i.e., the basis vector whose first nonzero entry is equal to one. For comparability we also need to recover a unique representation of the hybrid codewords, which is analogously given by their basis matrix in RREF.
We summarize:

\begin{rem}
We will always recover the basis matrix of the respective codeword in reduced row echelon form.
\end{rem}

In Algorithm \ref{alg2} we describe a decoding algorithm for spread codes for the REC.
This is a special case of~\cite[Alg.~1]{ma11j}, where the error space has dimension~$0$.
To use this very simple decoding algorithm we must assume that the Desarguesian spread code is of the form $\mathcal S_q(m,k,P^\top)$ as in Definition~\ref{D-DesSpread}
and where~$P$ is as in~\eqref{e-P}.
Moreover, we assume that the received matrix is decodable, i.e., it contains a nonzero row (without erasures).

\begin{algorithm}
\begin{algorithmic}
\REQUIRE{a received matrix $R=(R_1\mid \dots \mid R_m)\in\F^{k'\times n}$, where $1\leq k'\leq k$ and $n=mk$}
\STATE{find a nonzero row $\mathbf r$ in $R$ and represent it via $\psi$ as $(r_1, \dots , r_m)\in \F_{q^k}^m$}
\STATE{set $\mu := \min_i \{i \mid {r}_i \neq 0\}$ }
\FOR{$i=1,\dots,m$}
\STATE{compute $v_i := {r}_\mu^{-1} {r}_i$ }
\ENDFOR
\RETURN{$(v_1 , \dots , v_m )$}
\end{algorithmic}
\caption{Decoding of Desarguesian spread codes $\mathcal S_q(m,k,P^\top)$ in the REC.}
\label{alg2}
\end{algorithm}

Note that in the first step we used the isomorphism~$\psi$ between~$\F_q^k$ and $\F_q[\alpha]$.
Together with the fact that the matrix descriptions are from $\F_q[P^\top]$, this particular isomorphism guarantees that the output of the algorithm is indeed independent of the row~$r$ that was picked.
For the latter also remember that in a spread code the intersection of two codewords is trivial, and therefore any nonzero vector of a codeword uniquely identifies that codeword.
For further details on the algorithm we refer to~\cite{ma11j}.

To derive the computational complexity of our decoding complexities we first observe that the map~$\psi$ simply rewrites vector coefficients as polynomial coefficients,
and therefore its computational cost can be neglected in the following decoding complexity analyses.


\begin{thm}\label{T-Alg2Compl}
The computational complexity of Algorithm \ref{alg2} is in $O(k^2m)=O(kn)$.
\end{thm}
\begin{proof}
The algorithm needs at most $k$ divisions over $\F_{q^k}$ inside each of the $m$ blocks for the normalization.
Such a division can be done with $O(k^2)$ operations over $\F_q$.
\end{proof}

Next, in Algorithm \ref{alg1}, we describe a decoding algorithm for Desarguesian spread codes in the CEC.
We assume that the spread code is of the usual form $\mathcal S_q(m,k,P)$ with companion matrix~$P$ as in \eqref{e-P}.
Furthermore, we assume that the received matrix is decodable, in the sense that one nonzero block
has no erasures and that all blocks have at most $k-1$ columns erased (see Theorem~\ref{lem:k-1}).
Recall that every block of a matrix in~$\cM$ is either zero or invertible;  see Definition~\ref{D-DesSpread}.
Thus the unerased columns of a nonzero matrix are nonzero.

\begin{algorithm}
\begin{algorithmic}
\REQUIRE{a received matrix $R=(R_1\mid \dots \mid R_m)\in(\F_q \cup \{?\})^{k\times n}$ with a nonzero block $R_*$ without erasures}
\STATE{compute $R_*^{-1}$}
\FOR{$i=1,\dots,m$}
\IF{$R_i$ is nonzero}
\STATE{find a nonzero column $\mathbf r^{(i)}$ in $R_i$ and let $j_i\in\{0,\ldots,k-1\}$ be the index of the column $\mathbf r^{(i)}$ in~$R_i$}
\STATE{compute $h_i:=\psi(R_*^{-1}\mathbf r^{(i)})$}
\STATE{compute $u_i := \alpha^{-j_i}h_i$}
\ELSE
\STATE{set $u_i=0$}
\ENDIF
\ENDFOR
\STATE{set $\mu := \min_i \{i \mid u_i \neq 0\}$ }
\FOR{$i=1,\dots,m$}
\STATE{compute $v_i := u_\mu^{-1} u_i$ }
\ENDFOR
\RETURN{$(v_1, \dots , v_m )$}
\end{algorithmic}
\caption{Decoding of Desarguesian spread codes $\mathcal S_q(m,k,P)$ in the CEC.}
\label{alg1}
\end{algorithm}


The correctness of Algorithm~\ref{alg1} follows from the proof of Theorem~\ref{lem:k-1}.
The proof shows that, after normalizing the matrix via the nonerased block, we just have to decode every block~$R_*^{-1}R_i$ inside $\F_q[P]$.
Denoting the $j$th column of the matrix $P^\ell$ by $(P^\ell)_{(j)}$ one has $\psi((P^\ell)_{(j)})=\alpha^{j+\ell}$ for all $j=0,\ldots,k-1$.
Since the matrices $R_*^{-1}R_i$ are partially erased matrices from $\F_q[P]$, one then easily derives that~$u_i$ as defined in the algorithm satisfies $\phi(u_i)=R_*^{-1}R_i$.
This shows explicitly how any matrix in~$\F_q[P]$ is fully determined by any of its columns (if we know the position of that column).
From all this we conclude that the output $(v_1,\ldots,v_m)$ represents the desired codeword $\rowsp(v_1,\ldots,v_m)=\rowsp(u_1,\ldots,u_m)$ via the isomorphism~\eqref{e-DesGrass}.

\begin{thm}
The computational complexity of Algorithm \ref{alg1} is in $O(k^2m + k^3)=O(kn + k^3)$.
\end{thm}
\begin{proof}
The algorithm needs the inversion of a $k\times k$-matrix and at most $m$ multiplications of this inverted matrix with a vector.
With Gaussian elimination the former needs $O(k^3)$ operations, and the latter needs  $O(k^2 m)$ operations with normal matrix multiplication.
Afterwards the algorithm performs at most $2$ divisions over $\F_{q^k}$ inside each of the $m$ blocks.
Since such a division can be done with $O(k^2)$ operations over $\F_q$, the statement follows.
\end{proof}

Thus, if $k\in O(m)$, the complexity orders of Algorithms \ref{alg2} and \ref{alg1} are the same. That is, in this case
the channel models are equivalent from a decoding complexity point of view when using spread codes.

\medskip

As a final comparison, we derive the complexity of decoding hybrid codes in the CEC.
As explained in \cite[Section VII]{sk13} general hybrid codes can be decoded by first decoding all symbol
erasures in the Reed-Solomon code and then decoding the dimension errors in the subspace code.
Since in our case we assume that no dimension errors occurred, we simply have to decode all received vectors
in the Reed-Solomon code. We obtain the following result.

\begin{prop}
$[n,k,n']$-hybrid codes in $\G$ in the CEC can be decoded with a computational complexity in $O(k^3 m^2)=O(kn^2)$.
\end{prop}
\begin{proof}
Each of the $k$ received vectors is a codeword of the respective $[n,n']$-Reed-Solomon code. Using Forney's algorithm for erasure decoding \cite{fo65} each vector can be decoded with $O(n^2)$ operations over $\F_q$.
Moreover, we need to bring the decoded vectors, written as rows in a matrix, into reduced echelon form. With Gaussian elimination this can be done with $O(k^2 n)$ operations, which is negligible in the overall complexity order, since $k<n$.
\footnote{The complexity order of Forney's algorithm can be improved to $O(n\log^2 n \log \log n)$ by using simultaneous polynomial evaluation (see \cite[p.\ 216]{ro06b}). Then the overall complexity order of decoding hybrid codes becomes $O(kn\log^2 n \log \log n + k^2 n)$.}
\end{proof}

We see that, from a decoding complexity point of view, in both channel types spread codes  are advantageous compared to hybrid codes.


\section{Simultaneous Deletions and Column Erasures}\label{sec:simul}

In this section we also allow some deletions (i.e., dimension losses) to happen in the CEC.
This makes sense, since even if we handle symbol erasures according to the CEC  model,
the receiver might observe deletions due to rank deficiencies in the random coefficients
chosen in the network nodes during transmission.

\subsection{Symbol Erasure Correction Capabilities}

As in Subsection \ref{S-SpreadSMN} the main tool for erasure correction is Lemma~\ref{L-RP}.
But we cannot directly apply Theorem~\ref{lem:k-1} because a row erasure affects all blocks.
However, as we show next it is enough that one block has no column erasures to retrieve the original codeword.

\begin{thm}\label{T-DelEr}
Let $n=mk$ and $\C=\mathcal S_q(m,k,P)\subseteq \G$ be a Desarguesian spread code.
Let $r\in\{0,\ldots,k-1\}$. Then
\begin{enumerate}
\item $\C$ can uniquely decode any~$r$ row deletions and $k-r-1$ column erasures.
\item $\C$ can uniquely decode any~$r$ row deletions and any $k-r-1$ column erasures per block if one nonzero
        block is not affected by column erasures.
\end{enumerate}
\end{thm}

\begin{proof}
For both statements we have to consider the following situation.
Let $\U=\rs(U)\in\mathcal S_q(m,k,P)$ be sent where $U=(U_1 \mid \dots \mid U_m)$ and $U_i\in \F_q[P]$.
Suppose we receive
$$R=\gamma(AU + E),$$
where~$\gamma$ is the column erasure operator as in~\eqref{e-SMN},
$A\in\F_q^{k\times k}$ of rank $k-r$ represents the network operations responsible for the $r$ row deletions, and $E\in \{0, ?\}^{k\times n}$ is the symbol erasure matrix for the erasures specified in the theorem.

We have to show that we can uniquely recover~$\U$.
This translates into the following problem.
Let $A, A'\in\F_q^{k\times k}$ of rank $k-r$ and $(U_1\mid\ldots\mid U_m),\,(V_1\mid\ldots\mid V_m)\in\cM$  (see Definition~\ref{D-DesSpread}) such that
\begin{equation}\label{e-XAYB}
  \gamma(A(U_1\mid\ldots\mid U_m)+E)=\gamma(A'(V_1\mid\ldots\mid V_m)+E).
\end{equation}
We have to show that $\rs(U_1\mid\ldots\mid U_m)=\rs(V_1\mid\ldots\mid V_m)$.
Note that because of the generality of $A,\,A'$ we may assume that the lower $r$ rows are equal to zero in the matrices on the left and right hand side of~\eqref{e-XAYB}.
We consider the two cases of the theorem separately.
\begin{enumerate}
\item
After rearranging the columns we may write~\eqref{e-XAYB} as
$  A(M\mid \bar U)  =   A'(M'\mid \bar V)$,
where $\bar U, \bar V \in\{?\}^{k\times(k-r-1)}$ are the to-be-erased columns.
Considering only the nonzero rows we obtain
\[
   \hat{A}M=\hat{A'}M',
\]
where $\hat{A},\,\hat{A'}$ are $(k-r)\times k$-submatrices of~$A,A'$, respectively.
Note that $\rk(\hat{A}\mid\hat{A'})=k-r$ since each matrix has rank~$k-r$.
Consider now the (left) kernel of the matrix $S:=\big(\begin{smallmatrix}\bar U\\-\bar V\end{smallmatrix}\big) \in \F_q^{2k\times (k-r-1)}$.
Its dimension is
$$\dim\ker S\geq 2k-(k-r-1)=k+r+1.$$
As a consequence, $\rs(\hat{A}\mid\hat{A'})$ and $\ker S$ intersect nontrivially.
Thus we may choose a nonzero $(\mathbf a,\mathbf a')$ in $\rs(\hat{A}\mid\hat{A'})\cap\ker S$.
Then $\mathbf a(M\mid \bar U)=\mathbf a'(M'\mid \bar V)$, and this is a nonzero vector in the intersection of $\rs(M\mid \bar U)$ and $\rs(M'\mid \bar V)$.
Since~$\mathcal S_q(m,k,P)$ is a spread, this shows that these row spaces are equal and hence
$\rs(U_1\mid\ldots\mid U_m)=\rs(V_1\mid\ldots\mid V_m)$, as desired.
\item
Without loss of generality let the first block be nonzero and unaffected by column erasures.
Considering only the first block and one more block, say the $i$th one, we arrive at the situation of Case~$1$.\ for the spread code $\mathcal S_q(2,k,P)$.
Thus we conclude that $\rs(U_1\mid U_i)=\rs(V_1\mid V_i)$ for $i=2,\dots, m$.
Now invertibility of~$U_1,\,V_1$ implies $V_1^{-1}(V_1\mid V_i)=U_1^{-1}(U_1\mid U_i)$, and thus
$(V_1\mid\ldots\mid V_m)=V_1U_1^{-1}(U_1\mid\ldots\mid U_m)$, which shows that $\rs(U_1\mid\ldots\mid U_m) = \rs(V_1\mid\ldots\mid V_m)$.
\qedhere
\end{enumerate}
\end{proof}

As in Section \ref{sec:spreads} the proof only establishes uniqueness of a codeword that matches the received word after the specified erasures and deletions. An explicit decoding algorithm will be given in Algorithm \ref{alg4}.

In the following corollary we count the number of correctable erasure patterns when using the CEC and $r$ deletions have occurred.
For simplicity we only count erasure patterns of size $(k-r)\times n$ instead of $k\times n$, since we can assume that the first $k-r$ rows of the received matrix~$R$
are a basis of the received space. In this case any symbol erasures in the last $r$ rows can also be tolerated, no matter what type of codes we use for transmission.

\begin{cor}\label{C-SMNcount-r}
Let $\mathcal S_q(m,k,P)\subseteq \G$ be  a Desarguesian spread code.
Suppose the matrix representation
\[
  U=(0\mid\ldots\mid0\mid I\mid B_{i+1}\mid \ldots\mid B_m)\in\cM
\]
of a codeword $\U \in \cS_q(m,k,P)$
is transmitted over the CEC and $r<k$ deletions (or row erasures) occurred.
Let~$\ell$ be the number of nonzero blocks~$B_t$ and $N_r:= \sum_{j=0}^{k-r-1} \binom{k}{j} (2^{k-r}-1)^j$. 
Then
$$e_\ell:=N_r^m \left(1-(\frac{N_r-1}{N_r})^{\ell+1}\right)$$
symbol erasure patterns $E\in \{0,?\}^{(k-r)\times n}$ can be uniquely decoded.
As a consequence, the average number of correctable symbol erasure patterns $E\in\{0,?\}^{(k-r)\times n}$ for $\cS_q(m,k,P)$ is (at least)
$$e_{\text{avg}}^{(r)} := \frac{N_r^{m}}{q^n-1}\left(q^n-\Big[\frac{(q^k-1)(N_r-1)}{N_r}+1\Big]^m\right) .$$
\end{cor}
\begin{proof}
In each block of size $(k-r)\times k$ that may be affected by erasures, we have $\sum_{j=0}^{k-r-1} \binom{k}{j} (2^{k-r}-1)^j$ possible erasure patterns that affect up to $k-r-1$ columns.
The rest of the proof is analogous to the one of Theorem \ref{T-SMNcount}, using Theorem \ref{T-DelEr}.2 instead of
Theorem~\ref{lem:k-1} for counting the correctable erasure patterns.
\end{proof}

We now compare the performance of spread codes to the one of hybrid codes in this setting.

\begin{lem}\label{lem-hyb2}
Let $\cH_r\subseteq \G$ be  a $[n,k,n']$-hybrid code constructed from a subspace code $\C_r \subseteq \mathcal G_q(k,n')$ of minimum subspace distance $2(r+1)$.
After $r<k$ deletions (or row erasures) the number of correctable symbol erasure patterns $E\in\{0,?\}^{(k-r)\times n}$ 
 is
$$e_{\mathcal H_r} := \sum_{j=0}^{n-n'} \binom{n}{j} (2^{k-r}-1)^j .$$
\end{lem}
\begin{proof}
By Lemma \ref{L-HybEras} the hybrid code $\cH_r$ is able to decode $r$ deletions and any $n-n'$ column erasures, which implies the statement. 
\end{proof}

If we take a lifted Gabidulin code (see \eqref{eq:lifting}) of minimum subspace distance $2(r+1)$ as the subspace code $\C_r \subseteq \mathcal G_q(k,n')$, the rate of the  $[n,k,n']$-hybrid code~$\cH_r$ is
\begin{equation*}
\frac{\log_q(|\cH_r |)}{nk} =  \frac{\log_q(|\C_r|)}{nk} = \left\{\begin{array}{ll} \frac{(n'-k)(k-r)}{nk},&\text{if } n'-k\geq k \\[.5ex]
                                                                   \frac{k(n'-k-r)}{nk} =\frac{n'-k-r}{n},&\text{if }n'-k<k\end{array}\right.    ,
\end{equation*}
under the assumption that $r < \min\{k,n'-k\}$ (otherwise we get a trivial code of dimension $0$).
Let us compare this to a spread code in $ \G$ of approximately the same rate.
By~\eqref{e-Rspread} the latter has approximate rate $(n-k)/nk$; thus we need
\begin{equation*}
    n' \approx\left\{\begin{array}{ll} \frac{n-k}{k-r} +k,&\text{if }n\geq k(k-r+1)\\
    [.5ex]  \frac{n-k}{k} +k +r ,&\text{if }  n< k(k-r+1)\end{array}\right.
\end{equation*}
to achieve approximately the same rate in both codes.

We conclude this subsection with an example comparing the performance of hybrid and spread codes in this setting.

\begin{ex}
We fix $k=3, r=1$ with variable $n\geq 9$. To achieve approximately the same rate in the hybrid code we need $n'\approx (n+3)/2$.
Moreover, for the spread codes we fix $q=2$, whereas for the hybrid codes we pick $q$ as the smallest prime power exceeding $n-1$.
We obtain the following data:
$$
\begin{array}{|c|c||c|c||c|c|}
\hline
   n & n'\phantom{\Big|}&
   \textnormal{rate spread} & \textnormal{rate hybrid} & e_{\text{avg}}^{(r)} \text{ (see Cor.~\ref{C-SMNcount-r})}& e_{\mathcal H_r}\text{ (see Lem.~\ref{lem-hyb2})}\\\hline
9 &6\phantom{\Big|}& 0.229 & 0.222 & 241 & 2620\\\hline
12 &7\phantom{\Big|}& 0.255 & 0.222 & 3068 & 239122\\
 &8\phantom{\Big|}&  & 0.278 &  & 46666 \\\hline
15 &9\phantom{\Big|}& 0.271 & 0.267 & 36736 & 4502215\\\hline
18 &10\phantom{\Big|}& 0.281 & 0.259 & 422707 &372581830\\
 &11\phantom{\Big|}&  & 0.296 &  &85485592\\\hline
  \end{array}
$$
One can see that at comparable rates the hybrid codes can correct more erasure patterns than the respective spread codes.\end{ex}


\subsection{Decoding} 

Next we describe a decoding algorithm for spread codes in the CEC, assuming that also~$r$ row deletions have occurred.
The algorithm has two steps: in the first step we decode the column erasures in a specific Gabidulin code, and in the second step we decode
the row erasures in the spread code.

For the decoding algorithm described next we assume that the spread is of the form $\mathcal S_q(m,k,P^\top)$  
and that the received matrix is decodable in the sense that (after the row deletions) one nonzero block has no column erasures
and that all blocks have at most $k-r-1$ columns erased.
In the \textbf{for}-loop of the algorithm we consider the row vector $G'\in\F_{q^k}^{k-r}$ as generator matrix of a Gabidulin code of $\F_{q^k}$-dimension~$1$ and rank distance $k-r$;
see Proposition~\ref{P-FPGab}.
We use a known column erasure decoder for Gabidulin codes, e.g.\ from \cite{ri04p}. As always, $n=mk$.

\begin{algorithm}
\begin{algorithmic}
\REQUIRE{a received matrix $R=(R_1\mid \dots \mid R_m) \in (\F_q\cup \{?\})^{k\times n}$ of rank $k-r$ with a nonzero block $R_*$ without column erasures}
\STATE{row reduce $R$ and denote the $k-r$ non-zero rows by $\bar R=(\bar R_1\mid \dots \mid \bar R_m) $}
\STATE{compute $G' = 
\bar\psi(\bar R_*) $}   
\FOR{$i=1,\dots,m$}
\STATE{decode $\bar R_i$ in the Gabidulin code 
with generator matrix $G'$; call the output $\bar B_i$}
\ENDFOR
\STATE{use Algorithm \ref{alg2} to decode the resulting matrix $(\bar B_1\mid \dots \mid \bar B_m)$ in the spread $\mathcal S_q(m,k,P^\top)$}
\end{algorithmic}
\caption{Decoding of Desarguesian spread codes $\mathcal S_q(m,k,P^\top)$ in the CEC.}
\label{alg4}
\end{algorithm}

The correctness of the algorithm follows from Proposition~\ref{P-FPGab} and the correctness of Algorithm~\ref{alg2}.

\begin{prop}
Denote by $f_{Gab}(k,k-r)$ the computational complexity order of the Gabidulin decoder used in the \textbf{for}-loop of Algorithm \ref{alg4}.
Then the computational complexity order of Algorithm \ref{alg4} is in $O(k^2n + m f_{Gab}(k,k-r))$.

Using one of the Gabidulin decoders of \cite{ri04p,si09p} one can achieve an overall complexity order in $O( k^4 m)=O(k^3 n)$.
\end{prop}
\begin{proof}
Using Gaussian elimination, the row reduction of $R$ can be done with $O(k^2n)$ operations over $\F_q$.
The computation of $G'$ is simply done by rewriting the at most $(k-r)k$ coefficients in~$\F_q$. Furthermore, we know from Theorem~\ref{T-Alg2Compl} that the last part, i.e.\ Algorithm \ref{alg2}, needs $O(kn)$ operations.
Since the \textbf{for}-loop is executed at most $m$ times, the overall complexity order $O(k^2n + m f_{Gab}(k,k-r))$ follows.
\end{proof}

Note that Algorithm \ref{alg4} can also be used to decode spread codes in the column erasure channel without any row deletions.
However, since the complexity order is worse than for Algorithm \ref{alg1} it is preferable to use Algorithm \ref{alg1}
(and the respective spread codes in the non-transposed form) if we assume that no row deletions happen during transmission.

\begin{ex}
Consider the spread code $\mathcal S_2(2,4,P^\top)$, where
\[
     P=\begin{pmatrix}0&0&0&1\\1&0&0&1\\0&1&0&0\\0&0&1&0\end{pmatrix},
\]
i.e., $\alpha^4=\alpha +1$.
We receive the matrix
\[
    R=\left(\!\!\begin{array}{cccc|cccc}1&0&0&1&1&1&1&?\\1&0&0&0&1&0&1&?\\1&0&0&1&1&1&1&?\\0&0&0&1&0&1&0&?\end{array}\!\!\right),
\]
from which we compute the row reduced basis matrix
$$        \bar R=\left(\!\!\begin{array}{cccc|cccc}1&0&0&0&1&0&1&?\\0&0&0&1&0&1&0&?\end{array}\!\!\right),
$$
i.e., $r=2$ deletions have occurred. The first block has no column erasures, thus
$$ R_* = \left(\!\!\begin{array}{cccc}1&0&0&0\\0&0&0&1\end{array}\!\!\right)  $$
and hence
$$  G' =  \begin{pmatrix} 1 & \alpha^3  \end{pmatrix}    .$$
We now decode the first and second block in the Gabidulin code with generator matrix $G'$ and get the codewords
$(1,\alpha^3)$ and $(1+\alpha^2+\alpha^3, \alpha) = \alpha^{13}G'$, respectively. This corresponds to the matrix
$$ (B_1\mid B_2)= \left(\!\!\begin{array}{cccc|cccc}1&0&0&0&1&0&1&1\\0&0&0&1&0&1&0&0 \end{array}\!\!\right)  ,  $$
which we then decode with the help of Algorithm \ref{alg2} to the spread codeword represented by $(1, 1+\alpha^2+\alpha^3) \in \mathcal G_{2^4}(1,2)$.
The corresponding matrix representation in $\cS_2(2,4,P^\top)$ in RREF is
\[
    \left(\!\!\begin{array}{cccc|cccc}1&0&0&0&1&0&1&1\\0&1&0&0&1&0&0&1\\0&0&1&0&1&0&0&0\\0&0&0&1&0&1&0&0  \end{array}\!\!\right)  .
\]

\end{ex}

\medskip

Analogously to Section \ref{sec:decoding}, as a final comparison, we derive the complexity of decoding hybrid codes in the CEC.
We now have to first decode all symbol
erasures in the Reed-Solomon code and then decode the dimension errors in the subspace code.
 We obtain the following result.

\begin{prop}
$[n,k,n']$-hybrid codes in $\G$ in the CEC can be decoded with a computational complexity in $ O(k {n'}^3) \subseteq O(k n^3)$.
\end{prop}
\begin{proof}
In \cite[Section VI]{sk13} a decoding algorithm for the hybrid codes described above is given. In there the respective complexity order is derived as $O(r n'^3)$, where $r <k $ is the number of dimension errors the hybrid code is able to correct.
\end{proof}

Therefore,  if $n' \in O(n)$ and we allow deletions in the column erasure channel,  spread codes are again preferable from a decoding complexity point of view.


\section{Conclusions}\label{sec:conclusion}
We compared the symbol erasure correction capability of spread codes in the row erasure channel and the column erasure channel, compared those to the erasure correction capability of hybrid codes in the column erasure channel, and also investigated the according decoding complexities.
The results show that, depending on the application and the given parameters, any of the three combinations 
 might be preferable.
\begin{itemize}
\item Generally, spread codes bear the advantage that they can be constructed over any finite field, whereas for hybrid codes one needs a field size of at least the length of the vectors to be transmitted.
\item Moreover, for both the row and the column erasure channel there exist very efficient decoding algorithms for spread codes. These algorithms have lower complexity order than known decoding algorithms for hybrid codes.
\item On the other hand, when using the column erasure channel (with or without deletions), hybrid codes can correct more symbol erasure patterns than spread codes of comparable rate.
\item Lastly, when using spread codes, it depends on the parameters whether the row or the column erasure channel is the preferable model. For small $n$ (length of the vectors) compared to $k$ (dimension of the codewords), 
the row erasure channel performs better. However, for increasing $n$ the symbol erasure correction capability in the column erasure channel is exponentially larger than in the row erasure channel. The decoding complexity order for both models is comparable.
\end{itemize}
Overall, it depends on the importance of decoding speed and the field size in the given application if spread or hybrid codes are the better choice. However, for almost all parameter sets the column erasure channel will bear only advantages over the row erasure channel when considering a symbol erasure network channel. Since this channel model has not been studied very extensively yet (as opposed to the classical operator channel), this motivates future research for coding in the column erasure channel model.

\bibliographystyle{abbrv}
\bibliography{network_coding_stuff}

\begin{thebibliography}{10}

\bibitem{de78}
P.~Delsarte.
\newblock Bilinear forms over a finite field, with applications to coding
  theory.
\newblock {\em Journal of Combinatorial Theory, Series A}, 25(3):226--241,
  1978.

\bibitem{fo65}
G.~Forney.
\newblock On decoding {BCH} codes.
\newblock {\em IEEE Transactions on Information Theory}, 11(4):549--557, Oct
  1965.

\bibitem{ga85a}
E.~M. Gabidulin.
\newblock Theory of codes with maximum rank distance.
\newblock {\em Problemy Peredachi Informatsii}, 21(1):3--16, 1985.

\bibitem{ho17}
A.-L. Horlemann-Trautmann.
\newblock Message encoding and retrieval for spread and cyclic orbit codes.
\newblock {\em Designs, Codes and Cryptography}, Jun 2017.

\bibitem{ko08}
R.~K{\"o}tter and F.~R. Kschischang.
\newblock Coding for errors and erasures in random network coding.
\newblock {\em IEEE Transactions on Information Theory}, 54(8):3579--3591,
  2008.

\bibitem{li94}
R.~Lidl and H.~Niederreiter.
\newblock {\em Introduction to Finite Fields and their Applications}.
\newblock Cambridge University Press, Cambridge, London, 1994.
\newblock Revised edition.

\bibitem{ma08p}
F.~Manganiello, E.~Gorla, and J.~Rosenthal.
\newblock Spread codes and spread decoding in network coding.
\newblock In {\em Proceedings of the 2008 IEEE International Symposium on
  Information Theory (ISIT)}, pages 851--855, Toronto, Canada, 2008.

\bibitem{ma11j}
F.~Manganiello and A.-L. Trautmann.
\newblock Spread decoding in extension fields.
\newblock {\em Finite Fields and Applications}, 25:94--105, jan 2014.

\bibitem{ri04p}
G.~Richter and S.~Plass.
\newblock Fast decoding of rank-codes with rank errors and column erasures.
\newblock In {\em IEEE International Symposium on Information Theory (ISIT)},
  pages 398--398, 2004.

\bibitem{ro06b}
R.~Roth.
\newblock {\em Introduction to Coding Theory}.
\newblock Cambridge University Press, New York, NY, USA, 2006.

\bibitem{si17}
V.~Sidorenko, H.~Bartz, and A.~Wachter-Zeh.
\newblock Interleaved subspace codes in fountain mode.
\newblock In {\em 2017 IEEE International Symposium on Information Theory
  (ISIT)}, pages 799--803, June 2017.

\bibitem{si09p}
D.~Silva and F.~R. Kschischang.
\newblock Fast encoding and decoding of {G}abidulin codes.
\newblock In {\em IEEE International Symposium on Information Theory (ISIT)},
  pages 2858--2862, June 2009.

\bibitem{sk13}
V.~Skachek, O.~Milenkovic, and A.~Nedi\'c.
\newblock Hybrid noncoherent network coding.
\newblock {\em IEEE Transactions on Information Theory}, 59(6):3317--3331,
  2013.

\end{thebibliography}

\end{document}